\documentclass{amsart}

\usepackage{amsmath}
\usepackage{amsthm}
\usepackage{mathrsfs}
\usepackage{amssymb}
\usepackage{graphicx}
\usepackage[lined,vlined,boxed]{algorithm2e}

\newtheorem{theorem}{Theorem}

\newtheorem{proposition}[theorem]{Proposition} 
\newcommand{\BlackBox}{\rule{1.5ex}{1.5ex}}  
\renewenvironment{proof}{\par\noindent{\bf Proof\ }}{\hfill\BlackBox\\[2mm]}

\begin{document}
\pagestyle{empty}
\thispagestyle{empty}
\LinesNumbered
\IncMargin{1mm}

\begin{center}
{\large \bf Prismatic Algorithm for Discrete D.C.\@ Programming Problem}\\
\vspace{5mm}
Yoshinobu Kawahara and Takashi Washio\\
The Institute of Scientific and Industrial Research (ISIR)\\
Osaka University\\
8-1 Mihogaoka, Ibaraki-shi, Osaka 567-0047 JAPAN \\
\texttt{kawahara@ar.sanken.osaka-u.ac.jp} \\
\vspace{8mm}
\end{center}

\begin{center}
{\bf Abstract}:
\end{center}
In this paper, we propose the first exact algorithm for minimizing the difference of two submodular functions (D.S.), {\em i.e.}, the discrete version of the D.C.\@ programming problem. The developed algorithm is a branch-and-bound-based algorithm which responds to the structure of this problem through the relationship between submodularity and convexity. The D.S.\@ programming problem covers a broad range of applications in machine learning because this generalizes the optimization of a wide class of set functions. We empirically investigate the performance of our algorithm, and illustrate the difference between exact and approximate solutions respectively obtained by the proposed and existing algorithms in feature selection and discriminative structure learning.
\vspace{8mm}

\section{Introduction}
Combinatorial optimization techniques have been actively applied to many machine learning applications, where submodularity often plays an important role to develop algorithms \cite{HJZL06,KMGG08,TCG+09,KNTB09,KC10,NKI10,Bac10}. In fact, many fundamental problems in machine learning can be formulated as submoular optimization. One of the important categories would be the D.S.\@ programming problem, {\em i.e.}, the problem of minimizing the difference of two submodular functions. This is a natural formulation of many machine learning problems, such as learning graph matching \cite{CMC+09}, discriminative structure learning \cite{NB05}, feature selection \cite{Bac10} and energy minimization \cite{RMBK06}.

In this paper, we propose a prismatic algorithm for the D.S.\@ programming problem, which is a branch-and-bound-based algorithm responding to the specific structure of this problem. To the best of our knowledge, this is the first exact algorithm to the D.S.\@ programming problem (although there exists an approximate algorithm for this problem \cite{NB05}). As is well known, the branch-and-bound method is one of the most successful frameworks in mathematical programming and has been incorporated into commercial softwares such as CPLEX \cite{Iba87,HT96}. We develop the algorithm based on the analogy with the D.C.\@ programming problem through the continuous relaxation of solution spaces and objective functions with the help of the Lov\'{a}sz extension \cite{Lov83,HPTV91,Mur03}. The algorithm is implemented as an iterative calculation of binary-integer linear programming (BILP).

Also, we discuss applications of the D.S.\@ programming problem in machine learning and investigate empirically the performance of our method and the difference between exact and approximate solutions through feature selection and discriminative structure-learning problems.

The remainder of the paper is organized as follows.
In Section \ref{se:app}, we give the formulation of the D.S.\@ programming problem and then describe its applications in machine learning. In Section~\ref{se:algo}, we give an outline of the proposed algorithm for this problem. Then, in Section~\ref{se:basic}, we explain the details of its basic operations. And finally, we give several empirical examples using artificial and real-world datasets in Section~\ref{se:exper}, and conclude the paper in Section~\ref{se:concl}.

\subsubsection*{Preliminaries and Notation:}
A set function $f$ is called submodular if $f(A)+f(B)\geq f(A\cup B)+f(A\cap B)$ for all $A,B\subseteq N$, where $N=\{1,\cdots,n\}$ \cite{Edm70,Fuj05}. Throughout this paper, we denote by $\hat{f}$ the Lov\'{a}sz extension of $f$, {\em i.e.}, a continuous function $\hat{f}:\mathbb{R}^n\rightarrow\mathbb{R}$ defined by
\begin{equation*}
\hat{f}(\boldsymbol{p}) = {\textstyle \sum_{j=1}^{m-1}}(\hat{p}_j-\hat{p}_{j+1})f(U_j)+\hat{p}_m f(U_m),
\end{equation*}
where $U_j=\{i\in N:p_i\geq\hat{p}_j\}$ and $\hat{p}_1>\cdots>\hat{p}_m$ are the $m$ distinct elements of $\boldsymbol{p}$ \cite{Lov83,Mur03}. Also, we denote by $I_A\in\{0,1\}^n$ the characteristic vector of a subset $A\in N$, {\em i.e.}, $I_A=\sum_{i\in A}\boldsymbol{e}_i$ where $\boldsymbol{e}_i$ is the $i$-th unit vector. Note, through the definition of the characteristic vector, any subset $A\in N$ has the one-to-one correspondence with the vertex of a $n$-dimensional cube $D:=\{\boldsymbol{x}\in\mathbb{R}^n:0\leq x_i\leq 1 (i=1,\ldots,n)\}$. And, we denote by $(A,t)(T)$ all combinations of a real value plus subset whose corresponding vectors $(I_A,t)$ are inside or on the surface of a polytope $T\in\mathbb{R}^{n+1}$.

\section{The D.S.\@ Programming Problem and its Applications}
\label{se:app}

Let $f$ and $g$ are submodular functions. In this paper, we address an {\em exact} algorithm to solve the D.S.\@ programming problem, {\em i.e.}, the problem of minimizing the difference of two submodular functions:
\begin{equation}
\label{eq:dsprog}
\min_{A\in N}~~f(A)-g(A).
\end{equation}
As is well known, any real-valued function whose second partial derivatives are continuous everywhere can be represented as the difference of two convex functions \cite{HT96}. As well, the problem \eqref{eq:dsprog} generalizes a wide class of set-function optimization problems. Problem~\eqref{eq:dsprog} covers a broad range of applications in machine learning \cite{NB05,RMBK06,CMC+09,Bac10}. Here, we give a few examples.

\subsubsection*{Feature selection using structured-sparsity inducing norms}
Sparse methods for supervised learning, where we aim at finding good predictors from as few variables as possible, have attracted interest from machine learning community. This combinatorial problem is known to be a submodular maximization problem with cardinality constraint for commonly used measures such as least-squared errors  \cite{DK08,KNTB09}. And as is well known, if we replace the cardinality function with its convex envelope such as $l_1$-norm, this can be turned into a convex optimization problem. Recently, it is reported that submodular functions in place of the cardinality can give a wider family of polyhedral norms and may incorporate prior knowledge or structural constraints in sparse methods \cite{Bac10}. Then, the objective (that is supposed to be minimized) becomes the sum of a loss function (often, supermodular) and submodular regularization terms.

\subsubsection*{Discriminative structure learning}
It is reported that discriminatively structured Bayesian classifier often outperforms generatively one \cite{NB05,PB05}. One commonly used metric for discriminative structure learning would be EAR (explaining away residual) \cite{Bil00}. EAR is defined as the difference of the conditional mutual information between variables by class $C$ and non-conditional one, {\em i.e.}, $I(X_i;X_j|C)-I(X_i;X_j)$. In structure learning, we repeatedly try to find a subset in variables that minimize this kind of measure. Since the (symmetric) mutual information is a submodular function, obviously this problem leads the D.S.\@ programming problem \cite{NB05}.

\subsubsection*{Energy minimization in computer vision}
In computer vision, images is often modeled with a Markov random fields, where each node represents a pixel. Let $\mathcal{G}=(\mathcal{V},\mathcal{E})$ be the undirected graph, where a label $x_s\in\mathcal{L}$ is assigned on each node. Then, many tasks in computer vision can be naturally formulated in terms of energy minimization where the energy function has the form: $E(\boldsymbol{x})=\sum_{p\in\mathcal{V}}\theta_p(\boldsymbol{x}_p)+\sum_{(p,q)\in\mathcal{E}}\theta(\boldsymbol{x}_p,\boldsymbol{x}_q)$, where $\theta_p(i)$ and $\theta_{p,q}(i,j)$ are univariate and pairwise potentials. In a pairwise potential, submodularity is defined as $\theta_{pq}(x_p,x_q)+\theta_{pq}(x_p',x_q')\geq \theta_{pq}((x_p,x_q)\wedge(x_p',x_q'))+\theta_{pq}((x_p,x_q)\vee(x_p',x_q'))$ (see, for example, \cite{SKK+07}). Based on this, many energy function in computer vision can be written with a submodular function $E_1(\boldsymbol{x})$ and a supermodular function $E_2(\boldsymbol{x})$ as $E(\boldsymbol{x})=E_1(\boldsymbol{x})+E_2(\boldsymbol{x})$ (ex.\@ \cite{RMBK06}). Or, in case of binarized energy ({\em i.e.}, $\mathcal{L}=\{0,1\}$), even if such explicit decomposition is not known, a non-unique decomposition to submodular and supermodular functions can be always given \cite{She06}.

\section{Prismatic Algorithm for the D.S.\@ Programming Problem}
\label{se:algo}

By introducing an additional variable $t(\in\mathbb{R})$, Problem~\eqref{eq:dsprog} can be converted into the equivalent problem with a supermodular objective function and a submodular feasible set, {\em i.e.},
\begin{equation}
\label{eq:dsprog2}
\min_{A\in N,t\in\mathbb{R}}~t-g(A)~~~~~\text{s.t.}~~f(A)-t\leq 0.
\end{equation}

\begin{figure}[t]
\centering
\includegraphics[width=.4\linewidth]{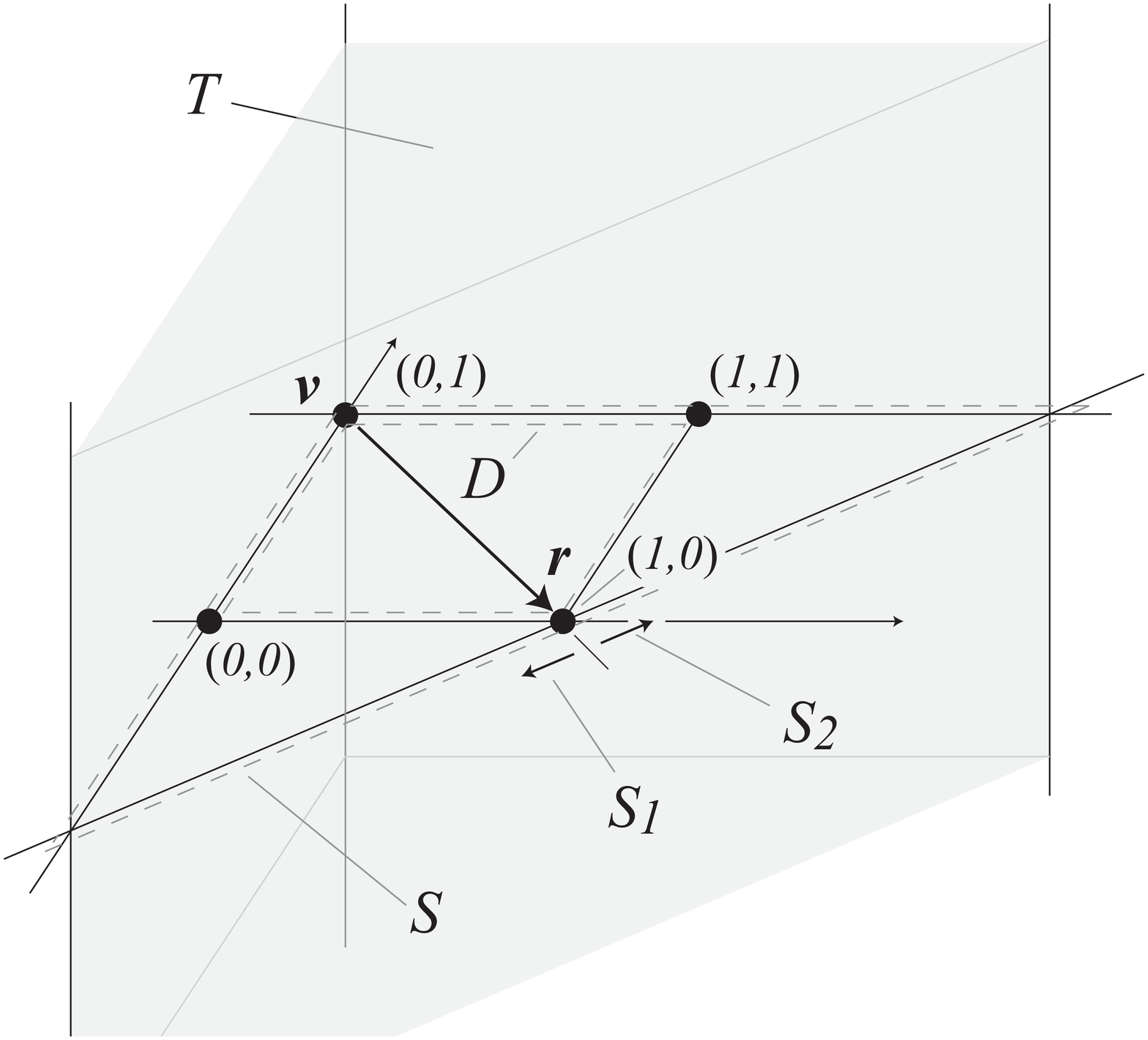}\vspace{-7mm}
\caption{Illustration of the prismatic algorithm for the D.S.\@ programming problem.}
\label{fig:prism}
\end{figure}

Obviously, if $(A^*,t^*)$ is an optimal solution of Problem~\eqref{eq:dsprog2}, then $A^*$ is an optimal solution of Problem~\eqref{eq:dsprog} and $t^*=f(A^*)$. The proposed algorithm is a realization of the branch-and-bound scheme which responds to this specific structure of the problem.

To this end, we first define a {\em prism} $T=T(S)\subset\mathbb{R}^{n+1}$ by
\begin{equation*}
T=\{(\boldsymbol{x},t)\in\mathbb{R}^n\times\mathbb{R}:\boldsymbol{x}\in S\},
\end{equation*}
where $S$ is an $n$-simplex. $S$ is obtained from the $n$-dimensional cube $D$ at the initial iteration (as described in Section~\ref{sss:first_prism}), or by the subdivision operation described in the later part of this section (and the detail will be described in Section~\ref{sss:subdivide}). The prism $T$ has $n+1$ edges that are vertical lines ({\em i.e.}, lines parallel to the $t$-axis) which pass through the $n+1$ vertices of $S$, respectively \cite{HPTV91}.

Our algorithm is an iterative procedure which mainly consists of two parts; {\em branching} and {\em bounding}, as well as other branch-and-bound frameworks \cite{Iba87}. In {\em branching}, subproblems are constructed by dividing the feasible region of a parent problem. And in {\em bounding}, we judge whether an optimal solution exists in the region of a subproblem and its descendants by calculating an upper bound of the subproblem and comparing it with an lower bound of the original problem. Some more details for branching and bounding are described as follows.

\subsubsection*{Branching}
The branching operation in our method is carried out using the property of a simplex. That is, since, in a $n$-simplex, any $r+1$ vertices are not on a $r-1$-dimensional hyperplane for $r\leq n$, any $n$-simplex can be divided as $S=\bigcup_{i=1}^p S_i$, where $p\geq 2$ and $S_i$ are $n$-simplices such that each pair of simplices $S_i, S_j (i\neq j)$ intersects at most in common boundary points (the way of constructing such partition is explained in Section~\ref{sss:subdivide}). Then, $T=\bigcup_{i=1}^p T_i$, where $T_i=\{(\boldsymbol{x},t)\in\mathbb{R}^n\times\mathbb{R}:\boldsymbol{x}\in S_i\}$, is a natural prismatic partition of $T$ induced by the above simplical partition.

\subsubsection*{Bounding}
For the bounding operation on $S$ (resp., $T$), we consider a polyhedral convex set $P$ such that $P\supset\tilde{D}$, where $\tilde{D}=\{(\boldsymbol{x},t)\in\mathbb{R}^n\times\mathbb{R}:\boldsymbol{x}\in D,\hat{f}(\boldsymbol{x})\leq t\}$ is the region corresponding to the feasible set of Problem~\eqref{eq:dsprog2}. At the first iteration, such $P$ is obtained as
\begin{equation*}
P_0 = \{(\boldsymbol{x},t)\in\mathbb{R}^n\times\mathbb{R}:\boldsymbol{x}\in S,t\geq\tilde{t}\},
\end{equation*}
where $\tilde{t}$ is a real number satisfying $\tilde{t}\leq\min\{f(A):A\in N\}$. Here, $\tilde{t}$ can be determined by using some existing submodular minimization solver \cite{Que98,FHI06}. Or, at later iterations, more refined $P$, such that $P_0\supset P_1\supset\cdots\supset\tilde{D}$, is constructed as described in Section~\ref{sss:outer}.

As described in Section~\ref{sss:lower}, a lower bound $\beta(T)$ of $t-g(A)$ on the current prism $T$ can be calculated through the binary-integer linear programming (BILP) (or the linear programming (LP)) using $P$, obtained as described above. Let $\alpha$ be the lowest function value ({\em i.e.}, an upper bound of $t-g(A)$ on $\tilde{D}$) found so far. Then, if $\beta(T)\geq\alpha$, we can conclude that there is no feasible solution which gives a function value better than $\alpha$ and can remove $T$ without loss of optimality.

The pseudo-code of the proposed algorithm is described in Algorithm~\ref{alg:prism}. 
In the following section, we explain the details of the operations involved in this algorithm.

\IncMargin{1mm}
\begin{algorithm}[t]
\SetKwFor{While}{while}{}{}
\SetKwIF{If}{ElseIf}{Else}{if}{then}{else if}{else}{endif}
Construct a simplex $S_0\supset D$, its corresponding prism $T_0$ and a polyhedral convex set $P_0\supset\tilde{D}$.\\
Let $\alpha_0$ be the best objective function value known in advance. Then, solve the BILP~\eqref{eq:mip} corresponding to $\alpha_0$ and $T_0$, and let $\beta_0=\beta(T_0,P_0,\alpha_0)$ and $(\bar{A}^0,\bar{t}_0)$ be the point satisfying $\beta_0=\bar{t}_0-g(\bar{A}^0)$.\\
Set $\mathcal{R}_0\leftarrow T_0$.\\
\While{$\mathcal{R}_k\neq\emptyset$}{
Select a prism $T_k^*\in\mathcal{R}_k$ satisfying $\beta_k=\beta(T_k^*),~(\bar{\boldsymbol{v}}^k,\bar{t}_k)\in T_k^*$.\\
\If{$(\bar{\boldsymbol{v}}^k,\bar{t}_k)\in\tilde{D}$}{Set $P_{k+1}=P_k$.}
\Else{Construct $l_k(\boldsymbol{x},t)$ according to \eqref{eq:cplane}, and set $P_{k+1}=\{(\boldsymbol{x},t)\in P_k:l_k(\boldsymbol{x},t)\leq 0\}$.}
Subdivide $T_k^*=T(S_k^*)$ into a finite number of subprisms $T_{k,j}$($j$$\in$$J_k$) (cf.~Section~\ref{sss:subdivide}).\\
For each $j\in J_k$, solve the BILP~\eqref{eq:mip} with respect to $T_{k,j}$, $P_{k+1}$ and $\alpha_k$.\\
Delete all $T_{k,j}$$(j$$\in$$J_k)$ satisfying (DR1) or (DR2). Let $\mathcal{R}'_k$ denote the collection of remaining prisms $T_{k,j}$$(j\in J_k)$, and for each $T\in\mathcal{M}_k'$ set
\begin{equation*}
\beta(T)=\max\{\beta(T_k^*),\beta(T,P_{k+1},\alpha_k)\}.
\end{equation*}\\
Let $F_k$ be the set of new feasible points detected while solving BILP in Step~11, and set
\begin{equation*}
\alpha_{k+1}=\min\{\alpha_k,\min\{t-g(A):(A,t)\in F_k\}\}.
\end{equation*}\\
Delete all $T$$\in$$\mathcal{M}_k$ satisfying $\beta(T)$$\geq$$\alpha_{k+1}$ and let $\mathcal{R}_k$ be $\mathcal{R}_{k-1}\setminus T_k \in\mathcal{M}_k$.\\
Set $\mathcal{M}_{k+1}$$\leftarrow$$(\mathcal{R}_k$$\setminus$$\{T_k^*\})$$\cup$$\mathcal{M}_k'$ and $\beta_{k+1}$$\leftarrow$$\min\{\beta(T)$$:$$T$$\in$$\mathcal{M}_{k+1}\}$.
}
\caption{Pseudo-code of the prismatic algorithm for the D.S\@ programming problem.}
\label{alg:prism}
\end{algorithm}

\section{Basic Operations}
\label{se:basic}
Obviously, the procedure described in Section~\ref{se:algo} involves the following basic operations:
\begin{enumerate}
\item Construction of the first prism: A prism needs to be constructed from a hypercube at first,
\item Subdivision process: A prism is divided into a finite number of sub-prisms at each iteration,
\item Bound estimation: For each prism generated throughout the algorithm, a lower bound for the objective function $t-g(A)$ over the part of the feasible set contained in this prism is computed,
\item Construction of cutting planes: Throughout the algorithm, a sequence of polyhedral convex sets $P_0,P_1,\cdots$ is constructed such that $P_0\supset P_1\supset\cdots\supset\tilde{D}$. Each set $P_j$ is generated by a cutting plane to cut off a part of $P_{j-1}$, and
\item Deletion of no-feasible prisms: At each iteration, we try to delete prisms that contain no feasible solution better than the one obtained so far.
\end{enumerate}

\subsection{Construction of the first prism}
\label{sss:first_prism}
The initial simplex $S_0\supset D$ (which yields the initial prism $T_0\supset\tilde{D}$) can be constructed as follows. Now, let $\boldsymbol{v}$ and $A_{\boldsymbol{v}}$ be a vertex of $D$ and its corresponding subset in $N$, respectively, {\em i.e.}, $\boldsymbol{v}=\sum_{i\in A_{\boldsymbol{v}}}\boldsymbol{e}_i$. Then, the initial simplex $S_0\supset D$ can be constructed by
\begin{equation*}
S_0 = \{\boldsymbol{x}\in\mathbb{R}^n:x_i\leq 1(i\in A_{\boldsymbol{v}}), x_i\geq 0(i\in N\setminus A_{\boldsymbol{v}}),\boldsymbol{a}^T\boldsymbol{x}\leq\gamma\},
\end{equation*}
where $\boldsymbol{a}=\sum_{i\in N\setminus A_{\boldsymbol{v}}}\boldsymbol{e}_i-\sum_{i\in A_{\boldsymbol{v}}}\boldsymbol{e}_i$ and $\gamma=|N\setminus A_{\boldsymbol{v}}|$. The $n+1$ vertices of $S_0$ are $\boldsymbol{v}$ and the $n$ points where the hyperplane $\{\boldsymbol{x}\in\mathbb{R}^n:\boldsymbol{a}^T\boldsymbol{x}=\gamma\}$ intersects the edges of the cone $\{\boldsymbol{x}\in\mathbb{R}^n:x_i\leq 1(i\in A_{\boldsymbol{v}}),x_i\geq 0(i\in N\setminus A_{\boldsymbol{v}})\}$. Note this is just an option and any $n$-simplex $S\supset D$ is available.

\subsection{Sub-division of a prism}
\label{sss:subdivide}
Let $S_k$ and $T_k$ be the simplex and prism at $k$-th iteration in the algorithm, respectively. We denote $S_k$ as $S_k=[\boldsymbol{v}_k^i,\ldots,\boldsymbol{v}_k^{n+1}]:=\text{conv}\{\boldsymbol{v}_k^1,\ldots,\boldsymbol{v}_k^{n+1}\}$ which is defined as the convex full of its vertices $\boldsymbol{v}_k^1,\ldots,\boldsymbol{v}_k^{n+1}$. Then, any $\boldsymbol{r}\in S_k$ can be represented as
\begin{equation*}
\boldsymbol{r}={\textstyle \sum_{i=1}^{n+1}}\lambda_i\boldsymbol{v}_k^i,~{\textstyle \sum_{i=1}^{n+1}\lambda_i=1},~\lambda_i\geq 0~(i=1,\ldots,n+1).
\end{equation*}
Suppose that $\boldsymbol{r}\neq\boldsymbol{v}_k^i~(i=1,\ldots,n+1)$. For each $i$ satisfying $\lambda_i>0$, let $S_k^i$ be the subsimplex of $S_k$ defined by
\begin{equation}
\label{eq:S}
S_k^i=[\boldsymbol{v}_k^1,\ldots,\boldsymbol{v}_k^{i-1},\boldsymbol{r},\boldsymbol{v}_k^{i+1},\ldots,\boldsymbol{v}_k^{n+1}].
\end{equation}
Then, the collection $\{ S_k^i:\lambda_i>0\}$ defines a partition of $S_k$, {\em i.e.}, we have \cite{HT96}
\begin{equation*}
{\textstyle\bigcup_{\lambda_i>0}}S_k^i=S_k,~\text{int}~S_k^i\cap\text{int}~S_k^j=\emptyset~~\text{ for }~~i\neq j.
\end{equation*}
In a natural way, the prisms $T(S_k^i)$ generated by the simplices $S_k^i$ defined in Eq.~\eqref{eq:S} form a partition of $T_k$. This subdivision process of prisms is {\em exhaustive}, {\em i.e.}, for every nested (decreasing) sequence of prisms $\{T_q\}$ generated by this process, we have $\bigcap_{q=0}^\infty T_q=\tau$, where $\tau$ is a line perpendicular to $\mathbb{R}^n$ (a vertical line) \cite{HPTV91}. Although several subdivision process can be applied, we use a classical {\em bisection} one, {\em i.e.}, each simplex is divided into subsimplices by choosing in Eq.~\eqref{eq:S} as
\begin{equation*}
\boldsymbol{r}=(\boldsymbol{v}_k^{i_1}+\boldsymbol{v}_k^{i_2})/2,
\end{equation*}
where $\|\boldsymbol{v}_k^{i_1}-\boldsymbol{v}_k^{i_2}\|=\max\{\|\boldsymbol{v}_k^i-\boldsymbol{v}_k^j\|:i,j\in\{0,\ldots,n\},i\neq j\}$ (see Figure~\ref{fig:prism}).

\subsection{Lower bounds}
\label{sss:lower}
Again, let $S_k$ and $T_k$ be the simplex and prism at $k$-th iteration in the algorithm, respectively. And, let $\alpha$ be an upper bound of $t-g(A)$, which is the smallest value of $t-g(A)$ attained at a feasible point known so far in the algorithm. Moreover, let $P_k$ be a polyhedral convex set which contains $\tilde{D}$ and be represented as
\begin{equation}
\label{eq:poly}
P_k = \{(\boldsymbol{x},t)\in\mathbb{R}^n\times\mathbb{R}:A_k\boldsymbol{x}+\boldsymbol{a}_kt\leq\boldsymbol{b}_k\},
\end{equation}
where $A_k$ is a real $(m\times n)$-matrix and $\boldsymbol{a}_k,\boldsymbol{b}_k\in\mathbb{R}^m$.\footnote{Note that $P_k$ is updated at each iteration, which does not depend on $S_k$, as described in Section~\ref{sss:outer}.} Now, a lower bound $\beta(T_k,P_k,\alpha)$ of $t-g(A)$ over $T_k\cap\tilde{D}$ can be computed as follows. In this section, we describe only the BILP implementation. The LP one and some empirical comparison are discussed in the supplementary document.

First, let $\boldsymbol{v}_k^i$ ($i=1,\ldots,n+1$) denote the vertices of $S_k$, and define $I(S_k)=\{i\in\{1,\ldots,n+1\}:\boldsymbol{v}_k^i\in \mathbb{B}^n\}$ and
\begin{equation*}
\mu = \left\{\begin{array}{ll}
\min\{\alpha,\min\{\hat{f}(\boldsymbol{v}_k^i)-\hat{g}(\boldsymbol{v}_k^i):i\in I(S)\}\}, & \text{if}~~I(S)\neq\emptyset, \\
\alpha, & \text{if}~~I(S)=\emptyset .
\end{array}\right.
\end{equation*}
For each $i=1,\ldots,n+1$, consider the point $(\boldsymbol{v}_k^i,t_k^i)$ where the edge of $T_k$ passing through $\boldsymbol{v}_k^i$ intersects the level set $\{(\boldsymbol{x},t):t-\hat{g}(\boldsymbol{x})=\mu\}$, {\em i.e.},
\begin{equation*}
t_k^i = \hat{g}(\boldsymbol{v}_k^i)+\mu~~(i=1,\ldots,n+1).
\end{equation*}
Then, let us denote the uniquely defined hyperplane through the points $(\boldsymbol{v}_k^i,t_k^i)$ by $H=\{(\boldsymbol{x},t)\in\mathbb{R}^n\times\mathbb{R}:\boldsymbol{p}^T\boldsymbol{x}-t=\gamma$, where $\boldsymbol{p}\in\mathbb{R}^n$ and $\gamma\in\mathbb{R}$. Consider the upper and lower halfspace generated by $H$, {\em i.e.}, $H_+=\{(\boldsymbol{x},t)\in\mathbb{R}^n\times\mathbb{R}:\boldsymbol{p}^T\boldsymbol{x}-t\leq\gamma\}$ and $H_-=\{(\boldsymbol{x},t)\in\mathbb{R}^n\times\mathbb{R}:\boldsymbol{p}^T\boldsymbol{x}-t\geq\gamma\}$. If $T_k\cap\tilde{D}\subset H_+$, then we see from the supermodularity of $g(A)$ (equivalently, the concavity of $\hat{g}(\boldsymbol{x})$) that
\begin{equation*}
{\small\begin{split}
\min\{t-g(A):(A,t)\in(A,t)(T_k\cap\tilde{D})\}
&>    \min\{t-g(A):(A,t)\in(A,t)(T_k\cap H_+)\} \\
&\geq \min\{t-\hat{g}(\boldsymbol{x}):(\boldsymbol{x},t)\in T_k\cap H_+\} \\
&=    \min\{t-\hat{g}(\boldsymbol{x}):(\boldsymbol{x},t)\in \{(\boldsymbol{v}_k^1,t_k^1),\ldots,(\boldsymbol{v}_k^{n+1},t_k^{n+1})\}\} = \mu.
\end{split}}
\end{equation*}
Otherwise, we shift the hyperplane $H$ (downward with respect to $t$) until it reaches a point $\boldsymbol{z}=(\boldsymbol{x}^*,t^*)$ ($\in T_k\cap P\cap H_-,\boldsymbol{x}^*\in\mathbb{B}^n$) ($(\boldsymbol{x}^*,t^*)$ is a point with the largest distance to $H$ and the corresponding pair $(A,t)$ (since $\boldsymbol{x}^*\in\mathbb{B}^n$) is in $(A,t)(T_k\cap P\cap H_-)$). Let $\bar{H}$ denote the resulting supporting hyperplane, and denote by $\bar{H}_+$ the upper halfspace generated by $\bar{H}$. Moreover, for each $i=1,\ldots,n+1$, let $\boldsymbol{z}^i=(\boldsymbol{v}_k^i,\bar{t}_k^i)$ be the point where the edge of $T$ passing through $\boldsymbol{v}_k^i$ intersects $\bar{H}$. Then, it follows $(A,t)(T_k\cap\tilde{D})\subset (A,t)(T_k\cap P)\subset (A,t)(T_k\cap\bar{H}_+)$, and hence
\begin{equation*}
{\small\begin{split}
\min\{t-g(A):(A,t)\in(A,t)(T_k\cap\tilde{D})\}
&\geq \min\{t-g(A):(A,t)\in(A,t)(T_k\cap\bar{H}_+)\} \\
&=    \min\{\bar{t}_k^i-\hat{g}(\boldsymbol{v}_k^i):i=1,\ldots,n+1\}.
\end{split}}
\end{equation*}
Now, the above consideration leads to the following BILP in $(\boldsymbol{\lambda},\boldsymbol{x},t)$:
\begin{equation}
\label{eq:mip}
\begin{split}
\max_{\boldsymbol{\lambda},\boldsymbol{x},t} ~~ \left({\textstyle\sum_{i=1}^{n+1}}t_i\lambda_i-t\right) ~~~~~~
\text{s.t.} &~~ A\boldsymbol{x}+\boldsymbol{a}t\leq\boldsymbol{b},~\boldsymbol{x}={\textstyle \sum_{i=1}^{n+1}}\lambda_i\boldsymbol{v}_k^i,~\boldsymbol{x}\in\mathbb{B}^n,\\
            &~~{\textstyle \sum_{i=1}^{n+1}}\lambda_i=1,~\lambda_i\geq 0~~(i=1,\ldots,n+1),
\end{split}
\end{equation}
where $A$, $\boldsymbol{a}$ and $\boldsymbol{b}$ are given in Eq.~\eqref{eq:poly}.

\begin{proposition}
\label{le:lower}
(a) If the system \eqref{eq:mip} has no solution, then intersection $(A,t)(T\cap\tilde{D})$ is empty.\\
(b) Otherwise, let $(\boldsymbol{\lambda}^*,\boldsymbol{x}^*,t^*)$ be an optimal solution of BILP~\eqref{eq:mip} and $c^*=\sum_{i=1}^{n+1}t_i\lambda_i^*-t^*$ its optimal value, respectively. Then, the following statements hold:\\
~(b1) If $c^*\leq 0$, then $(A,t)(T\cap\tilde{D})\subset (A,t)(H_+$).\\
~(b2) If $c^*>0$, then $\boldsymbol{z}=(\sum_{i=1}^{n+1}\lambda_i\boldsymbol{v}_k^i,t_k^*)$, $\boldsymbol{z}^i=(\boldsymbol{v}_k^i,\bar{t}_k^i)=(\boldsymbol{v}_k^i,t_k^i-c^*)$ and $\bar{t}_k^i-\hat{g}(\boldsymbol{v}_k^i)=\mu-c^* ~ (i=1,\ldots,n+1)$.
\end{proposition}
\begin{proof}
First, we prove part (a). Since every point in $S_k$ is uniquely representable as $\boldsymbol{x}=\sum_{i=1}^{n+1}\lambda_i\boldsymbol{v}^i$, we see from Eq.~\eqref{eq:poly} that the set $(A,t)(T_k\cap P)$ coincide with the feasible set of problem~\eqref{eq:mip}. Therefore, if the system~\eqref{eq:mip} has no solution, then $(A,t)(T_k\cap P)=\emptyset$, and hence $(A,t)(T_k\cap\tilde{D})=\emptyset$ (because $\tilde{D}\subset P$).\vspace*{1mm}\\
Next, we move to part (b). Since the equation of $H$ is $\boldsymbol{p}^T\boldsymbol{x}-t=\gamma$, it follows that determining the hyperplane $\bar{H}$ and the point $\boldsymbol{z}$ amounts to solving the binary integer linear programming problem:
\begin{equation}
\label{eq:mip2}
\max ~~ \boldsymbol{p}^T\boldsymbol{x}-t ~~~~~~
\text{s.t.} ~~ (\boldsymbol{x},t)\in T\cap P,~\boldsymbol{x}\in\mathbb{B}^n.
\end{equation}
Here, we note that the objective of the above can be represented as
\begin{equation*}
\boldsymbol{p}^T\boldsymbol{x}-t = \boldsymbol{p}^T\left({\textstyle \sum_{i=1}^{n+1}}\lambda_i\boldsymbol{v}_k^i\right)-t={\textstyle \sum_{i=1}^{n+1}}\lambda_i\boldsymbol{p}^T\boldsymbol{v}^i_k-t.
\end{equation*}
On the other hand, since $(\boldsymbol{v}^i,t_i)\in H$, we have $\boldsymbol{p}^T\boldsymbol{v}^i-t_i=\gamma$ ($i=1,\ldots,n+1$), and hence
\begin{equation*}
\boldsymbol{p}^T\boldsymbol{x}-t = {\textstyle \sum_{i=1}^{n+1}} \lambda_i(\gamma+t_i)-t={\textstyle \sum_{i=1}^{n+1}}t_i\lambda_i-t+\gamma.
\end{equation*}
Thus, the two BILPs \eqref{eq:mip} and \eqref{eq:mip2} are equivalent. And, if $\gamma^*$ denotes the optimal objective function value in Eq.~\eqref{eq:mip2}, then $\gamma^*=c^*+\gamma$. If $\gamma^*\leq\gamma$, then it follows from the definition of $H_+$ that $\bar{H}$ is obtained by a parallel shift of $H$ in the direction $H_+$. Therefore, $c^*\leq 0$ implies $(A,t)(T_k\cap P_k)\subset (A,t)(H_+)$, and hence $(A,t)(T_k\cap \tilde{D})\subset (A,t)(H_+)$.\vspace*{1mm}\\
Since $\bar{H}=\{(\boldsymbol{x},t)\in\mathbb{R}^n\times\mathbb{R}:\boldsymbol{p}^T\boldsymbol{x}-t=\gamma^*\}$ and $H=\{(\boldsymbol{x},t)\in\mathbb{R}^n\times\mathbb{R}:\boldsymbol{p}^T\boldsymbol{x}-t=\gamma\}$ we see that for each intersection point $(\boldsymbol{v}_k^i,\bar{t}_k^i)$ (and $(\boldsymbol{v}_k^i,t_k^i)$) of the edge of $T_k$ passing through $\boldsymbol{v}_k^i$ with $\bar{H}$ (and $H$), we have $\boldsymbol{p}^T\boldsymbol{v}_k^i-\bar{t}_k^i=\gamma^*$ and $\boldsymbol{p}^T\boldsymbol{v}_k^i-t_k^i=\gamma$, respectively. This implies that $\bar{t}_k^i=t_k^i+\gamma-\gamma^*=t_k^i-c^*$, and (using $t_k^i=\hat{g}(\boldsymbol{v}_k^i)+\mu$) that $\bar{t}_k^i=\hat{g}(\boldsymbol{v}_k^i)+\mu-c^*$.
\end{proof}

From the above, we see that, in the case (b1), $\mu$ constitutes a lower bound of $(t-g(A))$ wheres, in the case (b2), such a lower bound is given by $\min\{\bar{t}_k^i-\hat{g}(\boldsymbol{v}_k^i):i=1,\ldots,n+1\}$. Thus, Proposition~\ref{le:lower} provides the lower bound
\begin{equation}
\label{eq:beta}
\beta_k(T_k,P_k,\alpha) = \left\{\begin{array}{ll}
+\infty, & \text{if BILP \eqref{eq:mip} has no feasible point}, \\
\mu,     & \text{if}~ c^*\leq 0, \\
\mu-c^*  & \text{if}~ c^*>0.
\end{array}\right.
\end{equation}
As stated in Section~\ref{sss:delete}, $T_k$ can be deleted from further consideration when $\beta_k=\infty$ or $\mu$.

\subsection{Outer approximation}
\label{sss:outer}
The polyhedral convex set $P\supset\tilde{D}$ used in the preceding section is updated in each iteration, {\em i.e.}, a sequence $P_0, P_1,\cdots$ is constructed such that $P_0\supset P_1\supset\cdots\supset\tilde{D}$. The update from $P_k$ to $P_{k+1}$ ($k=0,1,\ldots$) is done in a way which is standard for pure outer approximation methods \cite{HT96}. That is, a certain linear inequality $l_k(\boldsymbol{x},t)\leq 0$ is added to the constraint set defining $P_k$, {\em i.e.}, we set
\begin{equation*}
P_{k+1} = P_k\cap\{(\boldsymbol{x},t)\in\mathbb{R}^n\times\mathbb{R}:l_k(\boldsymbol{x},t)\leq 0\}.
\end{equation*}
The function $l_k(\boldsymbol{x},t)$ is constructed as follows. At iteration $k$, we have a lower bound $\beta_k$ of $t-g(A)$ as defined in Eq.~\eqref{eq:beta} with $P=P_k$, and a point $(\bar{\boldsymbol{v}}_k,\bar{t}_k)$ satisfying $\bar{t}_k-\hat{g}(\bar{\boldsymbol{v}}_k)=\beta_k$. We update the outer approximation only in the case $(\bar{\boldsymbol{v}}_k,\bar{t}_k)\notin\tilde{D}$. Then, we can set
\begin{equation}
\label{eq:cplane}
l_k(\boldsymbol{x},t) = \boldsymbol{s}_k^T[(\boldsymbol{x},t)-\boldsymbol{z}_k]+(\hat{f}(\boldsymbol{x}_k^*)-t_k^*),
\end{equation}
where $\boldsymbol{s}_k$ is a subgradient of $\hat{f}(\boldsymbol{x})-t$ at $\boldsymbol{z}_k$. The subgradient can be calculated as, for example, stated in \cite{HK09} (see also \cite{Fuj05}).
\begin{proposition}
The hyperplane $\{(\boldsymbol{x},t)\in\mathbb{R}^n\times\mathbb{R}:l_k(\boldsymbol{x},t)=0\}$ strictly separates $\boldsymbol{z}_k$ from $\tilde{D}$, {\em i.e.}, $l_k(\boldsymbol{z}_k)>0$, and $l_k(\boldsymbol{x},t)\leq 0$ for $^\forall(\boldsymbol{x},t)\in\tilde{D}$.
\end{proposition}
\begin{proof}
Since we assume that $\boldsymbol{z}_k\notin\tilde{D}$, we have $l_k(\boldsymbol{z}_k)=(\hat{f}(\boldsymbol{x}_k^*)-t_k^*)$. And, the latter inequality is an immediate consequence of the definition of a subgradient.
\end{proof}

\subsection{Deletion rules}
\label{sss:delete}
At each iteration of the algorithm, we try to delete certain subprisms that contain no optimal solution. To this end, we adopt the following two deletion rules:
\begin{description}
\item[(DR1)] Delete $T_k$ if BILP~\eqref{eq:mip} has no feasible solution.
\item[(DR2)] Delete $T_k$ if the optimal value $c^*$ of BILP~\eqref{eq:mip} satisfies $c^*\leq0$.
\end{description}
The feasibility of these rules can be seen from Proposition~\ref{le:lower} as well as the D.C.\@ programing problem~\cite{HPTV91}. That is, (DR1) follows from Proposition~\ref{le:lower} that in this case $T\cap\tilde{D}=\emptyset$, {\em i.e.}, the prism $T$ is infeasible, and (DR2) from Proposition~\ref{le:lower} and from the definition of $\mu$ that the current best feasible solution cannot be improved in $T$.

\section{Experimental Results}
\label{se:exper}
We first provide illustrations of the proposed algorithm and its solution on toy examples from feature selection in Section~\ref{ss:art_data}, and then apply the algorithm to an application of discriminative structure learning using the UCI repository data in Section~\ref{ss:real}. The experiments below were run on a 2.8 GHz 64-bit workstation using Matlab and IBM ILOG CPLEX ver.~12.1.

\subsection{Application to feature selection}
\label{ss:art_data}

\begin{figure}[t]
\centering
\includegraphics[width=.32\linewidth]{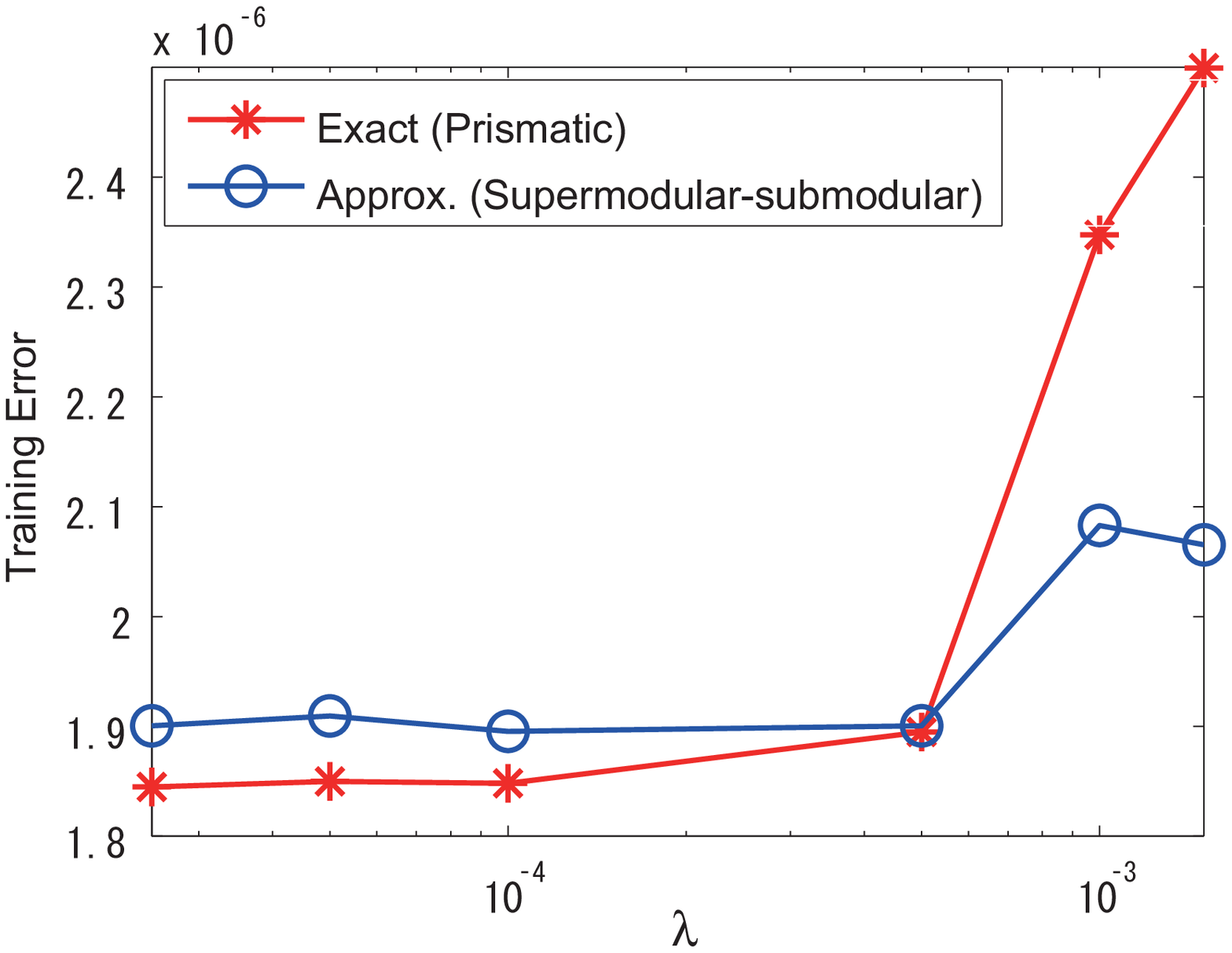}\hspace{.01\linewidth}
\includegraphics[width=.32\linewidth]{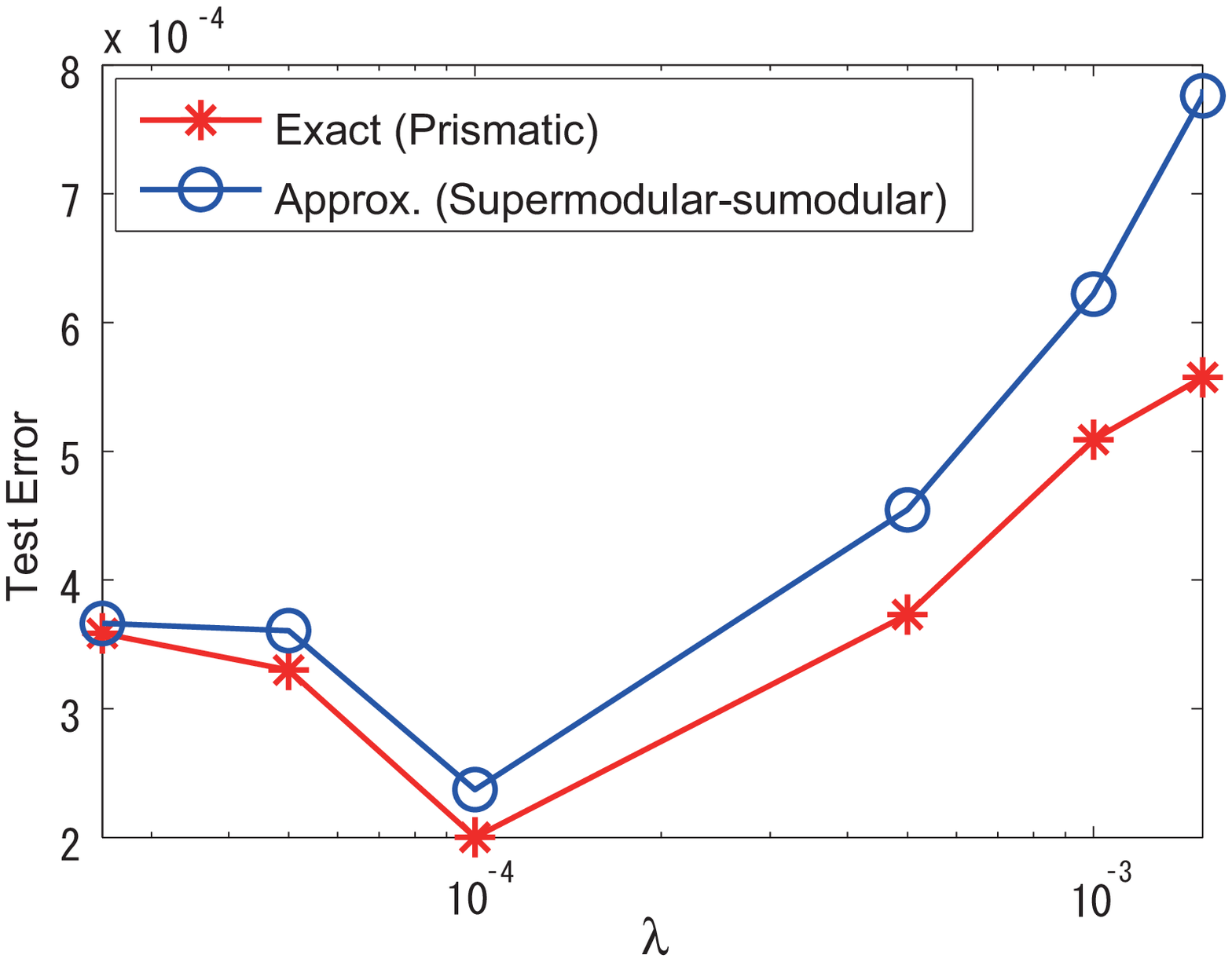}\hspace{.01\linewidth}
\includegraphics[width=.32\linewidth]{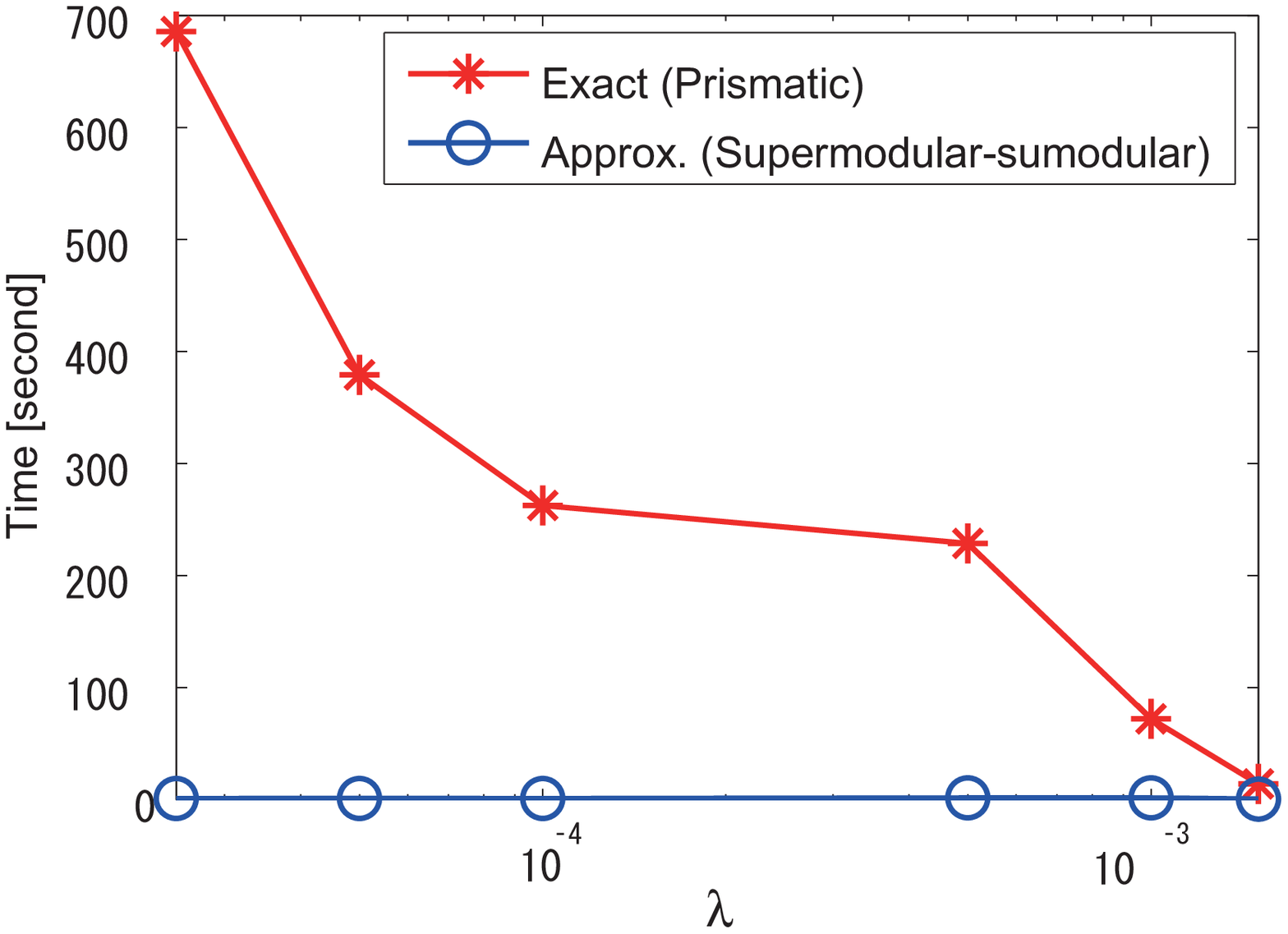}
\caption{Training errors, test errors and computational time versus $\lambda$ for the prismatic algorithm and the supermodular-sumodular procedure.}
\label{fig:fs}
\end{figure}

\begin{table}[t]
\vspace*{2mm}
\centering
\begin{tabular}{|ccc|cccc|}
\hline
p   & n   & k  & exact(PRISM) & SSP & greedy & lasso \\
\hline
120 & 150 & 5  & 1.8e-4 (192.6) & 1.9e-4 (0.93) & 1.8e-4 (0.45) & 1.9e-4 (0.78) \\
120 & 150 & 10 & 2.0e-4 (262.7) & 2.4e-4 (0.81) & 2.3e-4 (0.56) & 2.4e-4 (0.84) \\
120 & 150 & 20 & 7.3e-4 (339.2) & 7.8e-4 (1.43) & 8.3e-4 (0.59) & 7.7e-4 (0.91)\\
120 & 150 & 40 & 1.7e-3 (467.6) & 2.1e-3 (1.17) & 2.9e-3 (0.63) & 1.9e-3 (0.87)\\
\hline
\end{tabular}\vspace*{1mm}
\caption{Normalized mean-square prediction errors of training and test data by the prismatic algorithm, the supermodular-submodular procedure, the greedy algorithm and the lasso.}
\label{ta:fs_error}
\end{table}

We compared the performance and solutions by the proposed prismatic algorithm (PRISM), the supermodular-submodular procedure (SSP) \cite{NB05}, the greedy method and the LASSO. To this end, we generated data as follows:\@ Given $p$, $n$ and $k$, the design matrix $X\in\mathbb{R}^{n\times p}$ is a matrix of i.i.d.\@ Gaussian components. A feature set $J$ of cardinality $k$ is chosen at random and the weights on the selected features are sampled from a standard multivariate Gaussian distribution. The weights on other features are $0$. We then take $y=X\boldsymbol{w}+n^{-1/2}\|X\boldsymbol{w}\|_2\boldsymbol{\epsilon}$, where $\boldsymbol{w}$ is the weights on features and $\boldsymbol{\epsilon}$ is a standard Gaussian vector. In the experiment, we used the trace norm of the submatrix corresponding to $J$, $X_J$, {\em i.e.}, $\text{tr}(X_J^TX_J)^{1/2}$. Thus, our problem is $\min_{w\in\mathbb{R}^p}\frac{1}{2n}\|\boldsymbol{y}-X\boldsymbol{w}\|_2^2+\lambda\cdot\text{tr}(X_J^TX_J)^{1/2}$, where $J$ is the support of $\boldsymbol{w}$. Or equivalently, $\min_{A\in V}g(A)+\lambda\cdot\text{tr}(X_A^TX_A)^{1/2}$, where $g(A):=\min_{w_A\in\mathbb{R}^{|A|}}\|\boldsymbol{y}-X_A\boldsymbol{w}_A\|^2$. Since the first term is a supermodular function \cite{DK08} and the second is a submodular function, this problem is the D.S.\@ programming problem.

First, the graphs in Figure~\ref{fig:fs} show the training errors, test errors and computational time versus $\lambda$ for PRISM and SSP (for $p=120$, $n=150$ and $k=10$). The values in the graphs are averaged over 20 datasets. For the test errors, we generated another 100 data from the same model and applied the estimated model to the data. And, for all methods, we tried several possible regularization parameters. From the graphs, we can see the following: First, exact solutions (by PRISM) always outperform approximate ones (by SSP). This would show the significance of optimizing the submodular-norm. That is, we could obtain the better solutions (in the sense of prediction error) by optimizing the objective with the submodular norm more exactly. And, our algorithm took longer especially when $\lambda$ smaller. This would be because smaller $\lambda$ basically gives a larger size subset (solution). Also, Table~\ref{ta:fs_error} shows normalized-mean prediction errors by the prismatic algorithm, the supermodular-submodular procedure, the greedy method and the lasso for several $k$. The values are averaged over 10 datasets. This result also seems to show that optimizing the objective with the submodular norm exactly is significant in the meaning of prediction errors.

\subsection{Application to discriminative structure learning}
\label{ss:real}
Our second application is discriminative structure learning using the UCI machine learning repository.\footnote{\texttt{http://archive.ics.uci.edu/ml/index.html}} Here, we used CHESS, GERMAN, CENSUS-INCOME (KDD) and HEPATITIS, which have two classes. The Bayesian network topology used was the tree augmented naive Bayes (TAN) \cite{PB05}. We estimated TANs from data both in generative and discriminative manners. To this end, we used the procedure described in \cite{NB04} with a submodular minimization solver (for the generative case), and the one \cite{NB05} combined with our prismatic algorithm (PRISM) or the supermodular-submodular procedure (SSP) (for the discriminative case). Once the structures have been estimated, the parameters were learned based on the maximum likelihood method.

Table~\ref{ta:disc} shows the empirical accuracy of the classifier in [\%] with standard deviation for these datasets. We used the train/test scheme described in \cite{FGG97,PB05}. Also, we removed instances with missing values. The results seem to show that optimizing the EAR measure more exactly could improve the performance of classification (which would mean that the EAR is significant as the measure of discriminative structure learning in the sense of classification).\\

\begin{table}[t]
\centering
\begin{tabular}{|lcc|ccc|}
\hline
Data          & Attr.\@ & Class & exact (PRISM) & approx.\@ (SSP) & generative \\
\hline
Chess         & 36      & 2     & 96.6 ($\pm$0.69) & 94.4 ($\pm$0.71) & 92.3 ($\pm$0.79) \\
German        & 20      & 2     & 70.0 ($\pm$0.43) & 69.9 ($\pm$0.43) & 69.1 ($\pm$0.49) \\
Census-income & 40      & 2     & 73.2 ($\pm$0.64) & 71.2 ($\pm$0.74) & 70.3 ($\pm$0.74) \\
Hepatitis     & 19      & 2     & 86.9 ($\pm$1.89) & 84.3 ($\pm$2.31) & 84.2 ($\pm$2.11) \\
\hline
\end{tabular}\vspace*{1mm}
\caption{Empirical accuracy of the classifiers in [\%] with standard deviation by the TANs discriminatively learned with PRISM or SSP and generatively learned with a submodular minimization solver. The numbers in parentheses are computational time in seconds.}
\label{ta:disc}
\end{table}

\section{Conclusions}
\label{se:concl}
In this paper, we proposed a prismatic algorithm for the D.S.\@ programming problem~\eqref{eq:dsprog}, which is the first exact algorithm for this problem and is a branch-and-bound method responding to the structure of this problem. We developed the algorithm based on the analogy with the D.C.\@ programming problem through the continuous relaxation of solution spaces and objective functions with the help of the Lov\'{a}sz extension. We applied the proposed algorithm to several situations of feature selection and discriminative structure learning using artificial and real-world datasets.

The D.S.\@ programming problem addressed in this paper covers a broad range of applications in machine learning. In future works, we will develop a series of the presented framework specialized to the specific structure of each problem. Also, it would be interesting to investigate the extension of our method to enumerate solutions, which could make the framework more useful in practice.

\bibliography{dsprog_nips11}
\bibliographystyle{amsplain}
 
\end{document}